\documentclass[a4paper,11pt]{article}
\usepackage[ansinew]{inputenc}
\usepackage[spanish,english]{babel}
\usepackage{amsmath,amsfonts,amssymb,amsthm}
\usepackage{graphicx}
\usepackage{subfig}
\usepackage[usenames,dvipsnames]{color}
\usepackage{lineno}
\usepackage{enumerate}
\usepackage[colorlinks=true,citecolor=black,linkcolor=black,urlcolor=blue]{hyperref}
\usepackage{amscd,latexsym}
\usepackage{authblk}
\usepackage{cite}

\usepackage[T1]{fontenc}

\theoremstyle{plain}
\newtheorem{theorem}{Theorem}

\newtheorem{cor}{Corollary}

\newtheorem*{invariant*}{Invariant}

\date{}

\title{A note on non-crossing path partitions in the  plane\thanks{Partially supported by grant PID2023-150725NB-I00, funded by MICIU/AEI/10.13039/5011 00011033, and by grant E41-23R, funded by Gobierno de Arag\'on.} }

\author{J. Tejel}
\affil{Departamento de M\'etodos Estad\'\i sticos and IUMA, \\ Universidad de Zaragoza, Spain. \\ {\tt jtejel@unizar.es}}

\begin{document}

\maketitle

\begin{abstract}
In the paper ``Lower bounds on the number of crossing-free subgraphs of $K_N$'' ({\it Computational Geometry} 16 (2000), 211-221),
it is shown that a double chain of $n$ points in the plane admits at least $\Omega(4.642126305^n)$ polygonizations, and it is claimed that it admits at most $O(5.61^n)$ polygonizations. In this note, we provide a proof of this last result. The proof is based on counting non-crossing path partitions for points in the plane in convex position, where a non-crossing path partition consists of a set of paths connecting the points such that no two edges cross and isolated points are allowed.

We prove that a set of $n$ points in the plane in convex position admits $\mathcal{O}^*(5.610718614^{n})$ non-crossing path partitions and a double chain of $n$ points in the plane admits at least $\Omega(7.164102920^n)$ non-crossing path partitions. If isolated points are not allowed, we also show that there are $\mathcal{O}^*(4.610718614^n)$ non-crossing path partitions for $n$ points in the plane in convex position and at least $\Omega(6.164492582^n)$ non-crossing path partitions in a double chain of $n$ points in the plane.
In addition, using a particular family of non-crossing path partitions for points in convex position, we provide an alternative proof for the result that a double chain of $n$ points admits at least $\Omega(4.642126305^n)$  polygonizations.
\end{abstract}

\section{Introduction}

Given a set $S$ of $n$ points in the plane, a {\em non-crossing graph} of size $n$ is a graph whose vertices are the points of $S$, and whose edges are straight line segments that do not cross. A classical problem in Combinatorics
is counting (or bounding) the number of non-crossing graphs of a certain type $A$ (such as triangulations,  polygonizations, perfect matchings, ...) that $S$ admits. It is well-known that if $\Phi_A (S)$ denotes the number of non-crossing graphs of type $A$ that $S$ admits, and $\Phi_A (n)$ denotes the maximum number of non-crossing graphs of type $A$ over all sets $S$ of $n$ points in the plane, then $\Phi_A(n) \le c^{n}$, where $c$ is a constant~\cite{ajtai82}.

\begin{figure}[ht!]
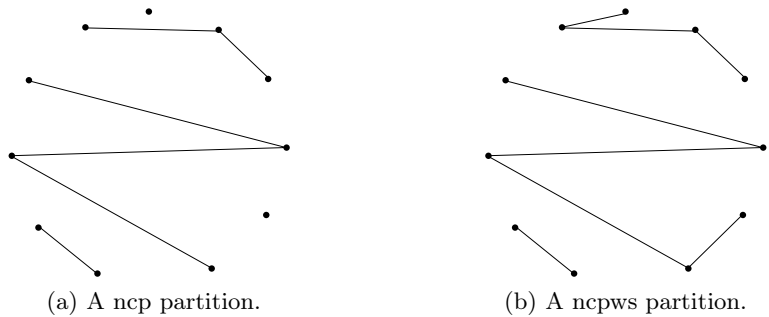

	\centering
	\subfloat[A ncp partition.]{
		\includegraphics[scale=0.4,page=18]{doublechain.pdf}
		\label{fig:ncp}
	}~~~~~~~~~~~~~~~~
	\subfloat[A ncpws partition.]{
		\includegraphics[scale=0.4,page=19]{doublechain.pdf}
		\label{fig:ncpws}
   }
\caption{Examples of ncp and ncpws partitions.}
\label{fig:partitions}
\end{figure}

Quite a lot of research has been done to estimate $\Phi_A(n)$, for different types of non-crossing graphs~\cite{asi18,esteban22,huemer15,garcia00,huemer19,rw23,sharir11,sharir12,sharir13}. For instance, in the case of the family $T$ of triangulations (non-crossing graphs where all faces are triangles except maybe for the outer face), the best-known bounds for $ \Phi_T(n)$ are$$\Omega(9.08^n)\le \Phi_T(n)\le O(30^n)$$ For the family $P$ of polygonizations (non-crossing Hamiltonian cycles), the best-known bounds for $ \Phi_P(n)$ are  $$\Omega(4.64^n)\le \Phi_P(n)\le O(54.55^n)$$ We refer the reader to~\cite{sheffer} for a summary of the best-known bounds up-to date of $\Phi_A(n)$, for several types of non-crossing graphs.

In~\cite{garcia00}, it is shown that $\Omega(4.642126305^n)\le \Phi_P(n)$, by proving that a particular configuration of $n$ points, the so-called double chain, contains at least $\Omega(4.642126305^n)$ polygonizations. In the same paper, it is claimed without proof that a double chain of $n$ points contains at most $O(5.61^n)$ polygonizations.

One of the purposes of this note is to provide a proof of such an upper bound, since we have been asked several times for a proof of this result, and it has been demonstrated to be useful for proving other results (for example, see~\cite{cruces25}). To this end, we will need to count the number of non-crossing path partitions for points in convex position. Thus, another purpose of this note is to give some insights into the number of non-crossing path partitions for some particular configurations of points in the plane.

Given a set $S$ of points in the plane, a {\em non-crossing path partition} of $S$ (ncp partition for short) is a non-crossing graph where each component is a path and where isolated vertices (singletons) are allowed. A singleton will also be considered as a path (of length 0). See Figure~\ref{fig:ncp} for an example. If singletons are not allowed, that is, if every point of $S$ must belong to a path of length at least 1, we will speak of {\em non-crossing path partitions without singletons} (ncpws partitions for short). See Figure~\ref{fig:ncpws}. We will refer to $NCP$ and $NCPWS$ as the classes of ncp partitions and ncpws partitions, respectively.

In Section~\ref{sec:convexo}, we provide exact and asymptotic formulas to compute the number of ncp and ncpws partitions for sets of points in convex position. In particular, we prove that for a set $S$ of $n$ points in the plane in convex position, $\Phi_{NCP} (S)$ is $\mathcal{O}^*(5.610718614^{n})$\footnote{In the $\mathcal{O}^*$ notation we neglect polynomial factors and give only the dominating exponential term.}, and  $\Phi_{NCPWS} (S)$ is $\mathcal{O}^*(4.610718614^n)$.
Section~\ref{sec:doublechain} is devoted to studying the double chain. We show that a double chain of $n$ points in the plane admits at most $O(5.610718614^{n})$ polygonizations. We also show that a double chain of $n$ points admits at least $\Omega(7.164102920^n)$ ncp partitions and at least $\Omega(6.164492582^n)$ ncpws partitions. In addition, using the ideas given in Section~\ref{sec:convexo}, we provide an alternative proof for the result that a double chain of $n$ points admits at least $\Omega(4.642126305^n)$ polygonizations. Finally, in Section~\ref{sec:conclusion}, we draw some conclusions.

\section{Ncp and ncpws partitions for points in convex position}\label{sec:convexo}

One of the most studied configurations of points in the plane is that of points in convex position, because of its
properties. In particular, when counting non-crossing graphs, the exact position of the points does not matter. For several results concerning counting non-crossing graphs for points in convex position, the reader is referred to~\cite{esteban22,flajolet99,noy98} and the references therein. For instance, the number of triangulations on $n+2$ points in convex position is the Catalan number $C_n=\frac{1}{n+1}\binom{2n}{n}$, the number of non-crossing perfect matchings on $2n$ points in convex position is again the Catalan number $C_n$, and the number of non-crossing spanning trees on $n+1$ points in convex position is the generalized Catalan number $\frac{1}{2n+1}\binom{3n}{n}$.

Using recurrence formulas, in this section, we compute the number of ncp and ncpws partitions that a set of $n$ points in convex position admits. We also give the asymptotic behaviour of these numbers.

\begin{theorem}\label{thm:ncp partitions}
Let $S$ be a set of $n$ points in the plane in convex position. Then, the number of ncp partitions that $S$ admits is $\mathcal{O}^*(5.610718614^n)$.
\end{theorem}

\begin{proof}
To count the number of ncp partitions, we will distinguish whether a point is a singleton, is an endpoint of a path, or is in the middle of a path. We assume that the points are clockwise numbered from 1 to $n$. See Figure~\ref{fig:g}.

Let $g(n)$ be the number of ncp partitions for $n$ points in convex position. By inspection, it is not difficult to check that $g(1)=1$, $g(2)=2$ or $g(3)=7$. Given a point $i$, let $f(n)$ be the number of ncp partitions such that $i$ is an endpoint of a path. In the same way, let $h(n)$ be the number of ncp partitions such that $i$ is in the middle of a path. Note that $f(n)$ and $h(n)$ do not depend on the point $i$ that is chosen.  Clearly, $f(1)=0$, $f(2)=1$, $h(1)=0$, $h(2)=0$ and $h(3)=1$. By convenience, we define $g(0)=1$, $f(0)=0$ and $h(0)=0$.

In the following, we give recurrence formulas for $g(n), f(n)$, and $h(n)$. To compute $g(n)$, observe that point 1 can be a singleton, an endpoint of a path, or a point in the middle of a path. Thus, we have
\begin{equation}\label{ecu:uno}
  g(n) = g(n-1)+f(n)+h(n)
\end{equation}
for $n\ge 1$.

\begin{figure}[ht!]
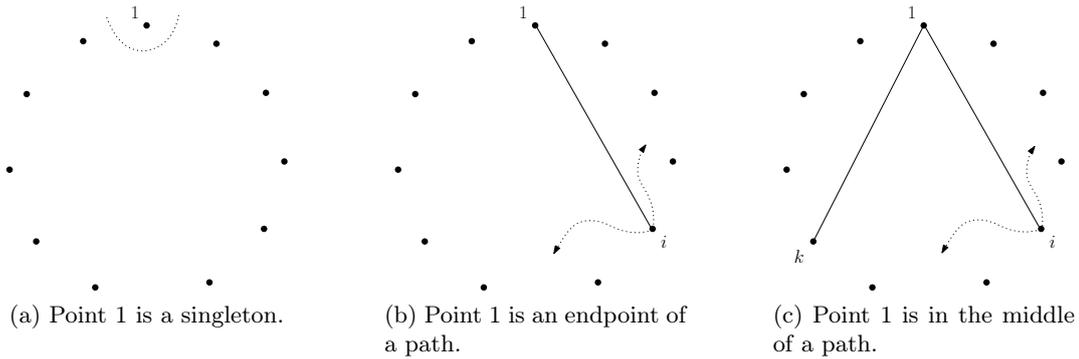

	\centering
	\subfloat[Point 1 is a singleton.]{
		\includegraphics[scale=0.4,page=5]{doublechain.pdf}
		\label{fig:singleton}
	}~~~~~~~~
	\subfloat[Point 1 is an endpoint of a path.]{
		\includegraphics[scale=0.4,page=6]{doublechain.pdf}
		\label{fig:endpoint}
	}~~~~~~~~
	\subfloat[Point 1 is in the middle of a path.]{
		\includegraphics[scale=0.4,page=7]{doublechain.pdf}
		\label{fig:middle}
   }
\caption{The three possibilities for a point in a ncp partition.}
\label{fig:g}
\end{figure}

To compute $f(n)$ (when 1 is an endpoint of a path), point 1 can be connected to a point $i$, for $i=2, \ldots , n$. After connecting 1 to $i$, we have three cases, depending on whether $i$ is the other endpoint of the path starting at 1, is connected to a point $j<i$, or is connected to a point $j>i$. See Figure~\ref{fig:endpoint}. In the first case, we have a ncp partition for the points from $2$ to $i-1$ (and we have $g(i-2)$ possible ncp partitions), and another ncp partition for the points from $i+1$ to $n$ (and we have $g(n-i)$ possible ncp partitions). In the second case, we have a ncp partition for the points from 2 to $i$, where $i$ is an endpoint of a path (and we have $f(i-1)$ possible ncp partitions), and a ncp partition for the points from $i+1$ to $n$ (and we have $g(n-i)$ possible ncp partitions). In the third case, we have a ncp partition for the points from $2$ to $i-1$, and a ncp partition for the points from $i$ to $n$, where $i$ is an endpoint of a path. Therefore, we have the following equality
\begin{equation}\label{ecu:dos}
  f(n) = \sum_{i=2}^n \big( g(i-2)g(n-i) + f(i-1)g(n-i)+g(i-2)f(n-i+1)\big)
\end{equation}
for $n\ge 2$. Equivalently,
\begin{equation}\label{ecu:tres}
  f(n) = \sum_{i=2}^n \big( g(i-2)g(n-i) + 2 f(i-1)g(n-i)\big)
\end{equation}
for $n\ge 2$.

To compute $h(n)$, we argue similarly. See Figure~\ref{fig:middle}. Point 1 is connected to a  point $i$, for $i=2, \ldots , n-1$, and to a point $k>i$. If $i$ is an endpoint of the path containing 1, then we have a ncp partition for the points from $2$ to $i-1$, and a ncp partition for the points from $i+1$ to $1$, where 1 is an endpoint of a path. If $i$ is connected to a point $j<i$, then we have a ncp partition for the points from $2$ to $i$, where $i$ is an endpoint of a path, and a ncp partition for the points from $i+1$ to $1$, where $1$ is an endpoint of a path. Last, if $i$ is connected to a point $j>i$, then, identifying points 1 and $i$ as a new point $1'$,  we have a ncp partition for the points from $2$ to $i-1$, and a ncp partition for the points from $i+1$ to $1'$, where $1'$ is in the middle of a path. Thus, the following equality holds
\begin{equation}\label{ecu:cuatro}
  h(n) = \sum_{i=2}^{n-1} \big( g(i-2)f(n-i+1) + f(i-1)f(n-i+1) + g(i-2) h(n-i+1)\big)
\end{equation}
for $n\ge 3$.

Let $G(z)=\sum_{n=0}^\infty g(n)z^n$, $F(z)=\sum_{n=0}^\infty f(n)z^n$ and $H(z)=\sum_{n=0}^\infty h(n)z^n$ be the generating functions of $g(n)$, $f(n)$ and $h(n)$, respectively. Using Equations~\ref{ecu:uno},~\ref{ecu:tres} and~\ref{ecu:cuatro}, after some calculations we obtain
\begin{eqnarray}
  G(z)-1 &=& zG(z)+F(z)+H(z) \\
  F(z) &=& z^2G(z)^2+2zF(z)G(z) \\
  H(z) &=& zG(z)F(z) + F(z)^2 + zG(z)H(z)
\end{eqnarray}

From these equations, we get that $G(z)$ must satisfy the following equation
\begin{equation}\label{ecu:cinco}
  (3z^3-4z^2)G(z)^3+(z^2+4z)G(z)^2+(-3z-1)G(z)+1=0
\end{equation}

A classical result from Bender~\cite{bender74} in analytic combinatorics establishes that if a generating function $G(z)$ satisfies $F(z,G(z))\equiv 0$ under some suitable conditions, then $g(n)=  \mathcal{O}^*((\frac{1}{r})^n)$, where $r=z$ and $s=w$ are the solutions of $F(z,w)=F_w(z,w)=0$, and $F_w$ is a partial derivative. Applying this result to Equation~\ref{ecu:cinco}, we obtain $r=0.178230289$ and $s=1.593329627$ when solving $F(z,w)=F_w(z,w)=0$, so $g(n) = \mathcal{O}^*((1/r)^n)=\mathcal{O}^*(5.610718614^n)$.
\end{proof}

Using the same approach, if singletons are not allowed, we can also count ncpws partitions for $n$ points in convex position.

\begin{theorem}\label{thm:ncpws partitions}
Let $S$ be a set of $n$ points in the plane in convex position. Then, the number of ncpws partitions that $S$ admits is $\mathcal{O}^*(4.610718614^n)$.
\end{theorem}

\begin{proof}
The proof follows the lines in the proof of Theorem~\ref{thm:ncp partitions}, so we omit some details.
Given $n$ points in convex position, let $gs(n)$, $fs(n)$ and $hs(n)$ be the number of ncpws partitions, the number of ncpws partitions when a point $i$ is an endpoint of a path, and the number of ncpws partitions when a point $i$ is in the middle of a path, respectively. The same reasoning used to obtain recurrence formulas for $g(n), f(n)$ and $h(n)$ yields the following recurrence formulas for $gs(n), fs(n)$ and $hs(n)$
\begin{equation}\label{ecu:nueve}
  gs(n) = fs(n)+hs(n)
\end{equation}
for $n\ge 1$ and $gs(0)=1$,
\begin{equation}\label{ecu:diez}
  fs(n) = \sum_{i=2}^n \big( gs(i-2)gs(n-i) + 2 fs(i-1)gs(n-i)\big)
\end{equation}
for $n\ge 2$ and $fs(0)=fs(1)=0$ and
\begin{equation}\label{ecu:once}
  hs(n) = \sum_{i=2}^{n-1} \big( gs(i-2)fs(n-i+1) + fs(i-1)fs(n-i+1) + gs(i-2) hs(n-i+1)\big)
\end{equation}
for $n\ge 3$ and $hs(0)=hs(1)=hs(2)=0$.

Note that the only difference of these equations concerning Equations~\ref{ecu:uno},~\ref{ecu:tres} and~\ref{ecu:cuatro} is that now $g(1)=0$ and the term $gs(n-1)$ disappears in Equation~\ref{ecu:nueve} because singletons are not allowed.

If $GS(z), FS(z)$ and $HS(z)$ are the generating functions of $gs(n), fs(n)$ and $hs(n)$, respectively, from Equations~\ref{ecu:nueve},~\ref{ecu:diez} and~\ref{ecu:once}, we obtain
\begin{eqnarray}
  GS(z)-1 &=& FS(z)+HS(z) \\
  FS(z) &=& z^2GS(z)^2+2zFS(z)GS(z) \\
  HS(z) &=& zGS(z)FS(z) + FS(z)^2 + zGS(z)HS(z)
\end{eqnarray}

From these equations, $GS(z)$ must satisfy the equation
\begin{equation}\label{ecu:quince}
  (z^3+4z^2)GS(z)^3+(-5z^2-4z)GS(z)^2+(4z+1)GS(z)-1=0
\end{equation}

Applying again Bender's result to Equation~\ref{ecu:quince}, we obtain $r=0.2168859312$ and $s=1.309350027$ when solving $F(z,w)=F_w(z,w)=0$. Therefore, $gs(n) = \mathcal{O}^*((1/r)^n)=\mathcal{O}^*(4.610718614^n)$.
\end{proof}

\begin{table}[!tb]
\centering
\begin{tabular}{|c|c|c|c|c|c|c|c|c|c|c|c|}
  \hline
  $n$ & 1 &  2 & 3 & 4 & 5 & 6 & 7 & 8 & 9 & 10 & 11 \\
  \hline
  $g(n)$ & 1 & 2 & 7 & 29 & 126 & 564 & 2591 & 12171 & 58237 & 282918 &  1391820  \\
  \hline
  $gs(n)$ & 0 & 1 & 3 & 10 & 35 & 128 & 483 & 1866 & 7344 & 29342 & 118701  \\ \hline
\end{tabular}
	\caption{$g(n)$ and $gs(n)$ for small values of $n$.}\label{tab:g}
\end{table}

Table~\ref{tab:g} summarizes $g(n)$ and $gs(n)$ for small values of $n$. We point out that $g(n)$ can be easily obtained from $gs(n)$, using the following reasoning. A ncp partition will contain $i$ singletons for some $i=0, \ldots , n$, and the remaining points will belong to a ncpws partition. Since the $i$ singletons can be chosen in $\binom{n}{i}$ ways, then $$g(n)=\sum_{i=0}^n \binom{n}{i} gs(n-i)$$

As $g(n)$ is the addition of $n+1$ terms, this allows us to compute the value $0 \le \alpha\le 1$ for which, when choosing $\alpha n$ singletons, the number of ncp partitions is $\mathcal{O}^*(5.610718614^n)$. To find the value of $\alpha$ that maximizes $\binom{n}{\alpha n}\cdot gs\left((1-\alpha)n\right)$, we use the fact that $gs\left((1-\alpha)n\right) = \mathcal{O}^*\left(4.610718614^{(1-\alpha)n}\right)$ and the property that $\binom{n}{\alpha n} \thickapprox 2^{H(\alpha)n}$, where $H(\alpha)$ is the binary entropy function defined as $H(\alpha)=-\alpha \log_2(\alpha)-(1-\alpha) \log_2(1-\alpha)$. The value of $\alpha$ that maximizes $2^{H(\alpha)n}\cdot 4.610718614^{(1-\alpha) n}$ is $\alpha=0.1782302$. Hence, when choosing approximately 18 percent of the points as singletons, the number of ncp partitions is approximately $\mathcal{O}^*(5.610718614^n)$.

We also point out that the sequence $0,$ $1,$ $3,$ $10,$ $35,$ $128,$ $483,$ $1866,$ $7344,$ $29342,$ $118701, \ldots $ of integers was already known. It is the sequence A303730 in the On-line Encyclopedia of Integer Sequences (OEIS)~\cite{A303730}. This sequence is described as ``the number of non-crossing path sets on $n$ nodes with each path having at least two nodes'', and was obtained by analyzing a special family of colored Motzkin paths of order 2~\cite{dejager21}. Thus, although a way of computing $gs(n)$
was already known, we provide here a different approach to computing it.

\section{Some results for the double chain}\label{sec:doublechain}

In this section, we revisit the problem of counting the number of polygonizations that a double chain $D$ consisting of $N$ points admits and prove that $$\Omega(4.642126305^N)\le \Phi_P(D)\le  O(5.610718614^N)$$ The lower bound was already proved in~\cite{garcia00}, but we give here an alternative proof. We also show lower bounds on the number of ncp and ncpws partitions that $D$ admits.

\begin{figure}[th!]
	\centering
	\subfloat[A polygonization on a double chain.]{
		\includegraphics[scale=0.4,page=1]{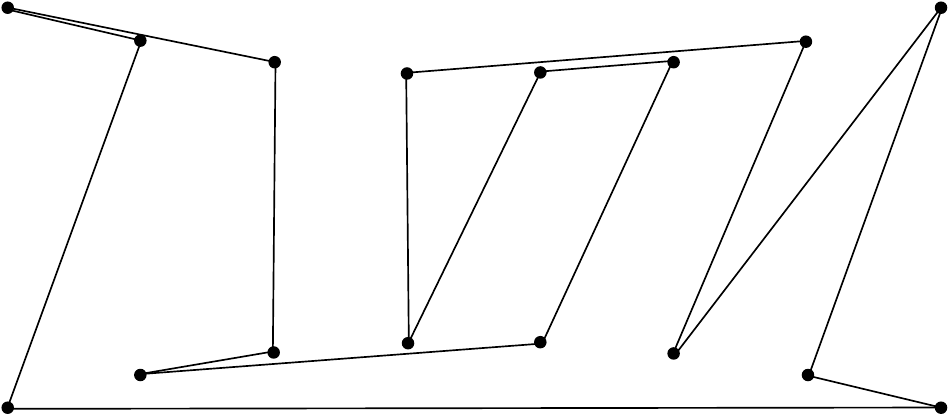}
		\label{fig:poligonization}
	}~~~~
	\subfloat[Two ncp partitions, one on $U$ and the other on $L$, obtained after removing the alternating edges.]{
		\includegraphics[scale=0.4,page=2]{doublechain.pdf}
		\label{fig:forest}
   }
\caption{Removing the alternating edges from a polygonization in a double chain, to obtain two ncp partitions.}
\label{fig:doublechain}
\end{figure}

We start by recalling what a double chain is. A {\em double chain} is a configuration of $N=2n$ points in the plane satisfying the following properties:
\begin{itemize}
\itemsep 0pt
  \item $n$ of the points are on a upper convex chain $U$, the other $n$ points are on a lower convex chain $L$, and both chains have opposed concavity.
  \item For every pair of points in $U$, the line passing through these points leaves all the points in $L$ below.
  \item For every pair of points in $L$, the line passing through these points leaves all the points in $U$ above.
\end{itemize}
The segments (edges) connecting points in opposite chains will be called {\em alternating} segments (edges). Figure~\ref{fig:poligonization} shows a double chain consisting of 16 points and a polygonization on it.

The double chain and its generalizations are a fundamental tool for obtaining lower bounds for $\Phi_A (n)$ and several types $A$ of non-crossing graphs~\cite{asi18, dumitrescu13, garcia00, huemer19, huemer15, rw23}. In the following theorem, we prove the bound claimed in~\cite{garcia00} that a double chain of $N$ points contains at most $O(5.610718614^N)$ polygonizations.

\begin{theorem}\label{thm:upperbound}
The number of polygonizations that a double chain $D$ of $N=2n$ points admits is at most $O(5.610718614^N)$.
\end{theorem}

\begin{proof}
Let $P$ be a polygonization of $D$. By removing the alternating edges of $P$, we obtain two ncp partitions, one on $U$ and the other on $L$. Since $P$ is a polygonization, the number of alternating edges is necessarily an even number. Thus, if $P$ contains $2k$ alternating edges, then each ncp partition consists of $k$ paths (recall that a singleton is a path of length 0). See Figure~\ref{fig:doublechain} for an illustration.

Suppose now that we are given two ncp partitions, one on $U$ and the other on $L$, each consisting of $k$ paths. To obtain a polygonization from the two ncp partitions, we must add alternating edges connecting the endpoints of the paths in $L$ to the endpoints of the paths in $U$. To avoid crossings between the alternating edges, the only option is that the endpoints on $U$ and the endpoints on $L$ are connected in order, from left to right. However, depending on the two ncp partitions, when adding the non-crossing alternating edges connecting the endpoints, we can obtain a polygonization or a set of non-crossing cycles. Figure~\ref{fig:doublechain_2} shows two ncp partitions and the two non-crossing cycles obtained after adding the non-crossing alternating edges. Note that in this example, one of the cycles consists of a duplicated edge.

\begin{figure}[th!]
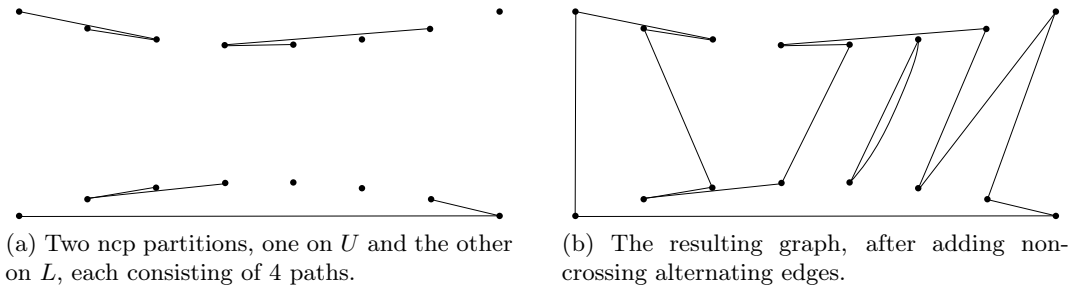

	\centering
	\subfloat[Two ncp partitions, one on $U$ and the other on $L$, each consisting of 4 paths.]{
		\includegraphics[scale=0.4,page=3]{doublechain.pdf}
		\label{fig:twoforests}
	}~~~~
	\subfloat[The resulting graph, after adding non-crossing alternating edges.]{
		\includegraphics[scale=0.4,page=4]{doublechain.pdf}
		\label{fig:twocycles}
   }
\caption{Adding non-crossing alternating edges to two ncp partitions.}
\label{fig:doublechain_2}
\end{figure}

As a consequence, if a polygonization $P_1$ defines the pair $(C_1, C'_1)$ of ncp partitions, $C_1$ on $U$ and $C'_1$ on $L$, and a polygonization $P_2$ defines the pair $(C_2, C'_2)$ of ncp partitions, $C_2$ on $U$ and $C'_2$ on $L$, then $(C_1, C'_1)$ and $(C_2, C'_2)$ must be different. Since the number of ncp partitions for $n$ points in convex position is $\mathcal{O}^*(5.610718614^n)$, then the number of polygonizations of $D$ is at most $O(5.610718614^N)$.
\end{proof}

We note that the previous theorem also holds, although the number of points on $U$ differs from the number of points on $L$. Thus, we have the following corollary.

\begin{cor}\label{cor:chain}
The number of polygonizations that a double chain $D$ of $n$ points on $U$ and $m$ points on $L$ admits is at most $O(5.610718614^{n+m})$.
\end{cor}

As explained in the proof of Theorem~\ref{thm:upperbound}, there is no guarantee that a polygonization is obtained by adding non-crossing alternating edges to two arbitrary ncp partitions. However, there are some cases where we can guarantee that.

Given a set $S$ of $n$ points in convex position numbered clockwise from 1 to $n$, an {\em ordered ncp partition} is a ncp partition on $S$ such that if $i$ and $j$ are the endpoints of a path, with $i<j$, then there are no other endpoints of paths between $i$ and $j$.

Assume now that from left to right, the points on $U$ ($L$) are numbered from 1 to $n$. Given two ordered ncp partitions, one on $L$ and the other on $U$, each consisting of $k$ paths (see Figure~\ref{fig:order_1}), if we skip the first endpoint on $U$ and the last endpoint on $L$, we always obtain a non-crossing Hamiltonian path when adding non-crossing alternating edges connecting the endpoints on $U$ to the endpoints on $L$ in order, from left to right. See Figure~\ref{fig:order_2} for an illustration of this construction. In addition, by adding a constant number of points (four in Figure~\ref{fig:order_3}), any non-crossing Hamiltonian path obtained as described previously can be transformed into a polygonization on a double chain. Thus, counting non-crossing Hamiltonian paths and counting polygonizations are asymptotically equivalent problems. As a consequence, if we can count the number of ordered ncp partitions consisting of $k$ paths for $n$ points in convex position, we can count a lot of non-crossing Hamiltonian paths (cycles) on a double chain.

\begin{figure}[ht!]
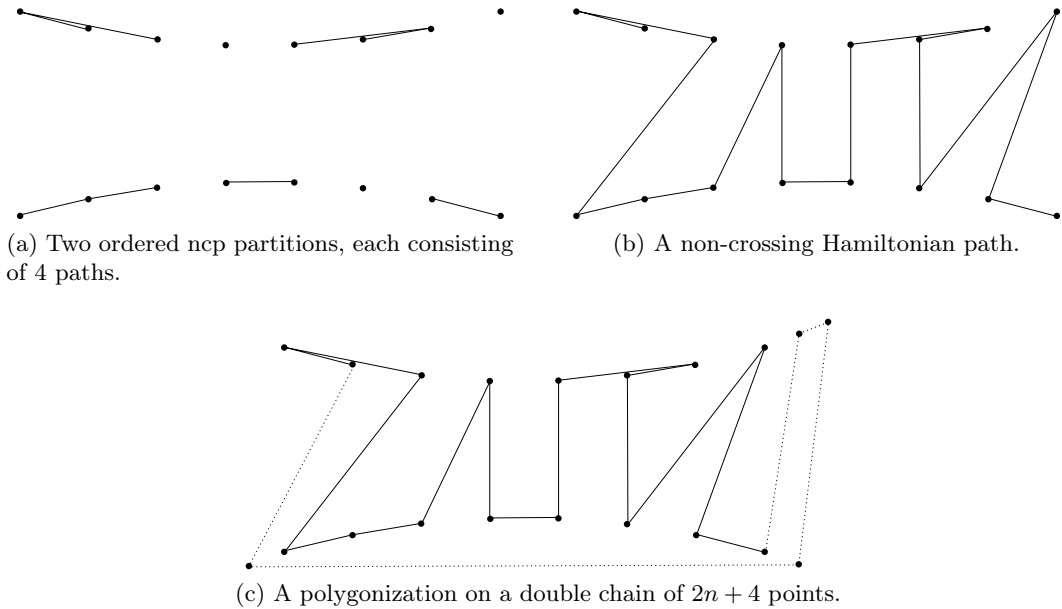

	\centering
	\subfloat[Two ordered ncp partitions, each consisting of 4 paths.]{
		\includegraphics[scale=0.4,page=8]{doublechain.pdf}
		\label{fig:order_1}
	}~~~~
	\subfloat[A non-crossing Hamiltonian path.]{
		\includegraphics[scale=0.4,page=9]{doublechain.pdf}
		\label{fig:order_2}
	}\\
	\subfloat[A polygonization on a double chain of $2n+4$ points.]{
		\includegraphics[scale=0.4,page=10]{doublechain.pdf}
		\label{fig:order_3}
   }
\caption{Obtaining a non-crossing Hamiltonian path (cycle), by adding non-crossing alternating edges to two ordered ncp partitions.}
\label{fig:order}
\end{figure}

The same happens in the proof given in~\cite{garcia00} to prove that a double chain of $n$ points contains at least $\Omega(4.642126305^n)$ polygonizations. That proof is based on counting curves of jump 1 visiting the points of a convex chain. These curves have the property that when the jumps are removed, the resulting graph is an ordered ncp partition for points in convex position. Thus, as a result of this observation, counting such curves is equivalent to counting a subset of ordered ncp partitions for points in convex position.

The following theorem counts the number of ordered ncp partitions for $n$ points in convex position.

\begin{theorem}\label{thm:ordered}
Let $S$ be a set of $n$ points in the plane in convex position. Then, the number of ordered ncp partitions that $S$ admits is $\mathcal{O}^*(4.642126305^n)$.
\end{theorem}

\begin{proof}
Again, the proof follows the lines in the proof of Theorem~\ref{thm:ncp partitions}, so we omit most of the details. Let $go(n)$, $fo(n)$, and $ho(n)$ be the number of ordered ncp partitions, the number of ordered ncp partitions when a point $i$ is an endpoint of a path, and the number of ordered ncp partitions when a point $i$ is in the middle of a path, respectively. The recurrence formulas for $go(n), fo(n)$ and $ho(n)$ are the following

\begin{equation}\label{ecu:16}
  go(n) = go(n-1)+fo(n)+ho(n)
\end{equation}
for $n\ge 1$ and $go(0)=1$,
\begin{equation}\label{ecu:17}
fo(n) = go(n-2)+fo(n-1) + \sum_{i=2}^n fo(i-1)go(n-i)
\end{equation}
for $n\ge 2$ and $fo(0)=fo(1)=0$, and
\begin{equation}\label{ecu:18}
ho(n) = \sum_{i=2}^{n-1} \big( go(i-2) fo(n-i+1) + fo(i-1)fo(n-i+1) + go(i-2) ho(n-i+1)\big)
\end{equation}
for $n\ge 3$ and $ho(0)=ho(1)=ho(2)=0$.

The first values of $go(n)$ are shown in Table~\ref{tab:g464}. Observe that in relation to Equations~\ref{ecu:uno},~\ref{ecu:tres} and~\ref{ecu:cuatro}, the main difference is in Equation~\ref{ecu:17}. Now, since the ncp partition is ordered, when 1 and $j$ are the endpoints of a path $p$, with $j\ge i$, then a path between 1 and $i$ cannot exist by definition. This implies that $i$ can only be point 2, and the terms $go(n-2)$ and $fo(n-1)$ appear in Equation~\ref{ecu:17}, depending on whether $i$ is the other endpoint of $p$ or not.

\begin{table}[!tb]
\centering
\begin{tabular}{|c|c|c|c|c|c|c|c|c|c|c|c|}
  \hline
  $n$ & 1 &  2 & 3 & 4 & 5 & 6 & 7 & 8 & 9 & 10 & 11 \\
  \hline
  $go(n)$ & 1 & 2 & 6 & 21 & 77 & 289 & 1107 & 4322 & 17162 & 69137 &  281917  \\
  \hline
\end{tabular}
	\caption{$go(n)$ for small values of $n$.}\label{tab:g464}
\end{table}

If $GO(z), FO(z)$ and $HO(z)$ are the generating functions of $go(n), fo(n)$ and $ho(n)$, respectively, from Equations~\ref{ecu:16},~\ref{ecu:17} and~\ref{ecu:18} we obtain

\begin{eqnarray}
  GO(z)-1 &=& zGO(z)+FO(z)+HO(z) \\
  FO(z) &=& zFO(z)GO(z)+z^2GO(z)+zFO(z) \\
  HO(z) &=& zFO(z)GO(z)+FO(z)^2+zGO(z)HO(z)
\end{eqnarray}
so $GO(z)$ must satisfy the equation
\begin{equation}\label{ecu:22}
(z^3-z^2)GO(z)^3+(2z^3-3z^2+2z)GO(z)^2+(z-1)GO(z)+z^2-2z+1
\end{equation}

Applying Bender's result to Equation~\ref{ecu:22}, we obtain $r=0.2154185247$ and $s=1.875110104$.
Therefore, $go(n) = \mathcal{O}^*((1/r)^n)=\mathcal{O}^*(4.642126305^n)$.
\end{proof}

We point out that ordered ncp partitions can be classified according to whether they consist of $1, 2, 3, \ldots, n$ paths.  Thus, there is an integer $k$ (which usually depends on $n$) such that the number of ordered ncp partitions consisting of $k$ paths, for $n$ points in convex position, is necessarily $\mathcal{O}^*(4.642126305^n)$. Therefore, a double chain of $N$ points admits at least $\mathcal{O}^*(4.642126305^N)$ polygonizations. In this way, we are providing an alternative proof of this result that was proved in~\cite{garcia00} using curves of jump 1.

To conclude this section, we show lower bounds for the number of ncp and ncpws partitions that a double chain admits.

\begin{theorem}\label{thm:doublechain}
Let $D$ be a double chain of $N=2n$ points. Then, $D$ admits at least $\Omega (7.164102920^N)$ ncp partitions and at least $\Omega(6.164492582^N)$ ncpws partitions.
\end{theorem}

\begin{proof}
We first show that $D$ admits at least $\Omega (6.164492582^N)$ ncpws partitions. To this end, we first count some ncpws partitions using only alternating edges. We assume that from left to right, the points of $U$ are numbered from 1 to $n$ and the points of $L$ from $n+1$ to $2n$.

We iteratively build some ncpws partitions using alternating edges as follows. In a generic step $i$, we build ncpws partitions using points from 1 to $i$ and points from $n+1$ to $n+i$ such that $i$ is connected to $n+i$ and both are the endpoints of the path connecting them, or $i$ is connected to $n+i$, but only one of them is an endpoint of a path. Let $A_i$ be the set of these ncpws partitions such that $i$ is connected to $n+i$ and both are the endpoints of the path connecting them, and let $B_i$ be the set of these ncpws partitions such that $i$ is connected to $n+i$, but only one of them is an endpoint of a path. In the first step, we consider only the ncpws partition that consists of the edge connecting $1$ and $n+1$. Hence, $A_1$ contains one ncpws partition and $B_1$ is empty.

\begin{figure}[ht!]
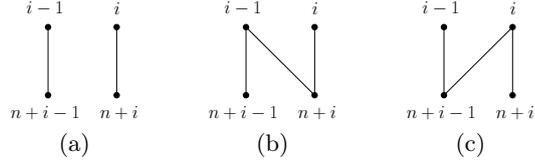

	\centering
	\subfloat[]{
		\includegraphics[scale=0.4,page=11]{doublechain.pdf}
		\label{fig:pair_1}
	}~~~~
	\subfloat[]{
		\includegraphics[scale=0.4,page=12]{doublechain.pdf}
		\label{fig:pair_2}
	}~~~~	
    \subfloat[]{
		\includegraphics[scale=0.4,page=13]{doublechain.pdf}
		\label{fig:pair_3}
   }
\caption{Building three new ncpws partitions when $i-1$ is connected to $n+i-1$ and both are endpoints.}
\label{fig:pair}
\end{figure}

\begin{figure}[ht!]
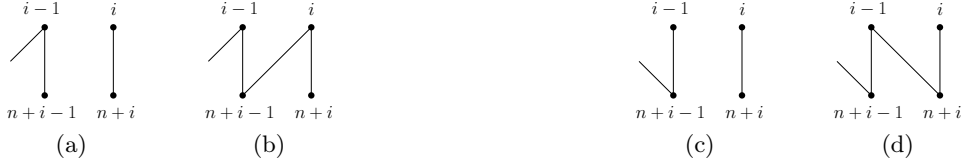

	\centering
	\subfloat[]{
		\includegraphics[scale=0.4,page=14]{doublechain.pdf}
		\label{fig:pair1_1}
	}~~~~
	\subfloat[]{
		\includegraphics[scale=0.4,page=15]{doublechain.pdf}
		\label{fig:pair1_2}
	}~~~~~~~~~~~~~~~~~~~~~~~~~~~~	
    \subfloat[]{
		\includegraphics[scale=0.4,page=16]{doublechain.pdf}
		\label{fig:pair1_3}
	}~~~~	
    \subfloat[]{
		\includegraphics[scale=0.4,page=17]{doublechain.pdf}
		\label{fig:pair1_4}
   }
\caption{Building two new ncpws partitions when $i-1$ is connected to $n+i-1$ but only one of them is an endpoint of a path.}
\label{fig:pair1}
\end{figure}

For $i=2, \ldots , n$, from a ncpws partition in $A_{i-1}$, we build three new ncp partitions as shown in Figure~\ref{fig:pair}. One of them (Figure~\ref{fig:pair_1}) belongs to $A_i$ and the other two belong to $B_i$. From a ncpws partition in $B_{i-1}$, we build two new ncpws partitions as shown in Figures~\ref{fig:pair1_1} and~\ref{fig:pair1_2} when $n+i-1$ is an endpoint, and in Figures~\ref{fig:pair1_3} and~\ref{fig:pair1_4} when $i-1$ is an endpoint.
Again, one of the two new ncpws partitions belongs to $A_i$ and the other to $B_i$.

For $i=1, \ldots , n$, we define $a_i=|A_i|$ and $b_i=|B_i|$. Note that $a_1=1$, $b_1=0$, $a_2=1$, and $b_2=2$. Thus, for $i=2, \ldots , n$, the following recurrences hold by construction

\begin{eqnarray*}
  a_i &=& a_{i-1}+b_{i-1} \\
  b_i &=& 2a_{i-1}+b_{i-1}
\end{eqnarray*}
After some easy calculations, these recurrences are equivalent to the following ones
\begin{eqnarray*}
  a_i &=& 2a_{i-1}+a_{i-2} \\
  b_i &=& 2b_{i-1}+b_{i-2}
\end{eqnarray*}
Using standard techniques to solve recurrences, we obtain the following formulas for $a_i$ and $b_i$
\begin{eqnarray*}
  a_i &=& c_1\left(1+\sqrt{2}\right)^i + c_2\left(1-\sqrt{2}\right)^i \\
  b_i &=& c_3\left(1+\sqrt{2}\right)^i + c_4\left(1-\sqrt{2}\right)^i
\end{eqnarray*}
where $c_1, c_2, c_3$ and $c_4$ are constants. When $n$ is large enough, the term $\left(1-\sqrt{2}\right)^n$ vanishes, so the behavior of $a_n$ and $b_n$ is asymptotically $O\left(\left(1+\sqrt{2}\right)^n\right)$. Hence, we can build $O\left(\left(1+\sqrt{2}\right)^{\frac{N}{2}}\right)$ ncpws partitions on $D$, using only alternating edges.

Now, we choose $\alpha n$ points of $L$ and other $\alpha n$ points of $U$, with $0\le \alpha \le 1$ to be determined. There are $\binom{n}{\alpha n}^2$ ways to choose these $2\alpha n$ points. With these $2\alpha n$ points, we build ncpws partitions using alternating edges, as previously described. With the remaining $(1-\alpha)n$ points of $L$ ($U$), we build a ncpws partition. Then, for any value of $\alpha$, the number of these ncpws partitions on $D$ is asymptotically
$$c(\alpha) = \binom{n}{\alpha n}^2 \left(4.610718614\right)^{2(1-\alpha)n}\left(1+\sqrt{2}\right)^{\frac{2\alpha n}{2}}$$
Using the fact that $\binom{n}{\alpha n} \thickapprox 2^{H(\alpha)n}$, elementary calculus shows that the maximum of $c(\alpha)$ is achieved at $\alpha = 0.25205209$, so $c(0.25205209)\approx 6.164492582^{2n}$. Therefore, $D$ admits at least $\Omega (6.164492582^N)$ ncpws partitions.

To prove that $D$ admits $\Omega (7.164102920^N)$ ncp partitions, we argue similarly. We choose $\alpha n$ points of $L$ and other $\alpha n$ points of $U$, with $0\le \alpha \le 1$. With these $2\alpha n$ points, we build ncpws partitions using alternating edges. With the remaining $(1-\alpha)n$ points of $L$ ($U$), we build a ncp partition. Then, for any value of $\alpha$, the number of these ncp partitions on $D$ is asymptotically
$$c(\alpha) = \binom{n}{\alpha n}^2 \left(5.610718614\right)^{2(1-\alpha)n}\left(1+\sqrt{2}\right)^{\frac{2\alpha n}{2}}$$
The maximum of $c(\alpha)$ is achieved at $\alpha = 0.2211799253$, so $c(0.2211799253)$ is approximately $7.164102920^{2n}$. Therefore, $D$ admits at least $\Omega (7.164102920^N)$ ncp partitions.
\end{proof}

\section{Conclusions}\label{sec:conclusion}

In this note, we have provided a proof for the result that a double chain of $n$ points admits at most $O(5.61^n)$ polygonizations, using the fact that the number of ncp partitions for $n$ points in convex position is $\mathcal{O}^*(5.610718614^{n})$. We have also proved several results concerning ncp and ncpws partitions, for points in convex position and points in a double chain.

\begin{figure}[ht!]
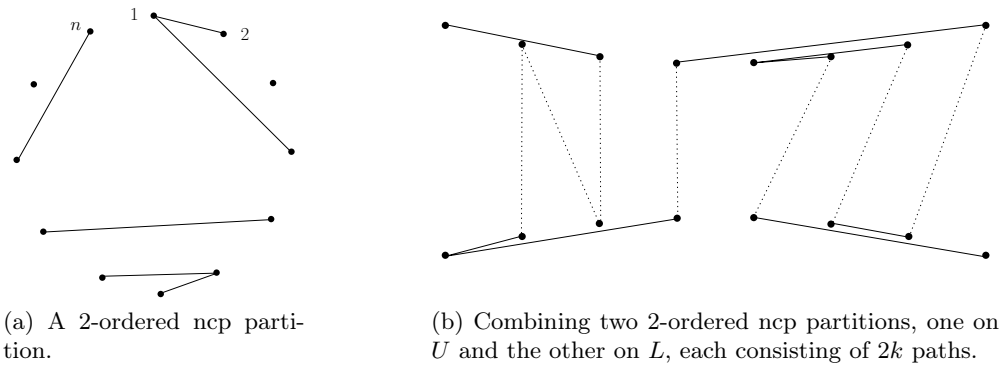

	\centering
	\subfloat[A 2-ordered ncp partition.]{
		\includegraphics[scale=0.4,page=20]{doublechain.pdf}
		\label{fig:twojumps_1}
	}~~~~~~~~~~~~
	\subfloat[Combining two 2-ordered ncp partitions, one on $U$ and the other on $L$, each consisting of $2k$ paths.]{
		\includegraphics[scale=0.45,page=21]{doublechain.pdf}
		\label{fig:twojumps_2}
   }
\caption{2-ordered ncp partitions.}
\label{fig:twojumps}
\end{figure}

Using ordered ncp partitions for points in convex position, we have shown how to build a large family of non-crossing Hamiltonian paths (cycles) on a double chain. One interesting question is whether other families of ncp partitions for points in convex position allow us to build large families of non-crossing Hamiltonian paths (cycles) on a double chain. For example, Figure~\ref{fig:twojumps_1} shows what we call a 2-ordered ncp partition. In a 2-ordered ncp partition, every path with endpoints $i<j$ satisfies the condition that it contains exactly one path between $i$ and $j$, or
is contained in another path. Choosing a 2-ordered ncp partition consisting of $2k$ paths on $L$, and another 2-ordered ncp partition consisting of $2k$ paths on $U$, it is not difficult to check that by skipping the first endpoint on $U$ and the last endpoint on $L$, a non-crossing Hamiltonian path can be obtained by adding alternating edges. Figure~\ref{fig:twojumps_2} illustrates this construction. We can prove that the number of 2-ordered ncp partitions for $n$ points in convex position is $\mathcal{O}^*(3.79^{n})$. Therefore, the number of non-crossing Hamiltonian paths on a double chain of $N$ points that are obtained using 2-ordered ncp partitions is also $\mathcal{O}^*(3.79^{N})$. Unfortunately, $3.79^{N}$ is not larger than $4.64^{N}$, the bound achieved using ordered ncp partitions for points in convex position.

Lastly, based on the results obtained for ncp and ncpws partitions, we hope that the insights provided in this note will be helpful in continuing this long line of research on counting non-crossing graphs on sets of points in the plane.

\bibliographystyle{plainurl}
\bibliography{references}

\end{document}